\documentclass{article}
\pdfoutput=1
\usepackage[T1]{fontenc}
\usepackage[utf8]{inputenc}
\usepackage{csquotes}
\usepackage[english]{babel}

\usepackage{fontaxes}
\usepackage{lmodern}

\usepackage{microtype}

\usepackage{bbm} 

\usepackage{bm}

\usepackage{graphicx}
\graphicspath{{img/}}

\usepackage{amssymb,amsmath,amsthm}

\usepackage{color}
\usepackage[dvipsnames]{xcolor}

\usepackage[%
    backend=biber,
    style=alphabetic,
    sorting=nyt,
    isbn=true,
    maxbibnames=10,
]{biblatex}
\usepackage{titling}
\usepackage[noblocks]{authblk}

\predate{}
\postdate{}

\DeclareSourcemap{%
	\maps[datatype=bibtex]{%
		\map[overwrite]{%
			\step[fieldsource=doi, final]
			\step[fieldset=url, null]
			\step[fieldset=eprint, null]
		}
	}
}

\usepackage{amssymb,amsmath,amsthm}

\usepackage{mathtools}
\mathtoolsset{centercolon}

\DeclareMathOperator{\tr}{Tr}

\DeclareMathOperator{\dom}{domain}
\DeclareMathOperator{\img}{image}
\DeclareMathOperator{\Ex}{\mathbf{E}}
\DeclareMathOperator{\Prob}{P}

\DeclareMathOperator{\entropyOp}{H}
\DeclareMathOperator{\vecOp}{vec}
\DeclareMathOperator{\loccOp}{LOCC}
\DeclareMathOperator{\stabOp}{Stab}
\DeclareMathOperator{\mincutOp}{min-cut}
\DeclareMathOperator{\nullityOp}{\nu}

\DeclarePairedDelimiter\bra{\langle}{\vert}
\DeclarePairedDelimiter\ket{\vert}{\rangle}

\DeclarePairedDelimiterX\ketbra[2]{\vert}{\vert}{#1 {\delimsize\rangle\langle} #2}
\DeclarePairedDelimiterX\braket[2]{\langle}{\rangle}%
{#1\,\delimsize\vert\,\mathopen{}#2}
\DeclarePairedDelimiterX\brakett[3]{\langle}{\rangle}%
{#1\,\delimsize\vert\,\mathopen{}#2\,\delimsize\vert\,\mathopen{}#3}
\DeclarePairedDelimiter\abs{\lvert}{\rvert}

\DeclarePairedDelimiter{\paren}{(}{)}
\DeclarePairedDelimiterXPP\bigo[1]{O}{(}{)}{}{#1}
\DeclarePairedDelimiterXPP\littleo[1]{o}{(}{)}{}{#1}
\DeclarePairedDelimiterXPP\bigtildeo[1]{\tilde O}{(}{)}{}{#1}
\DeclarePairedDelimiterXPP\bigomega[1]{\Omega}{(}{)}{}{#1}
\DeclarePairedDelimiterXPP\bigtheta[1]{\Theta}{(}{)}{}{#1}
\DeclarePairedDelimiterXPP\diagonal[1]{\diag}{(}{)}{}{#1}
\DeclarePairedDelimiterXPP\domain[1]{\dom}{(}{)}{}{#1}
\DeclarePairedDelimiterXPP\image[1]{\img}{(}{)}{}{#1}
\DeclarePairedDelimiterXPP\trace[1]{\tr}{[}{]}{}{#1}
\DeclarePairedDelimiterXPP\ptrace[2]{\tr_{#1}}{[}{]}{}{#2}
\DeclarePairedDelimiterXPP\probability[1]{\Prob}{[}{]}{}{#1}
\DeclarePairedDelimiterXPP\expectation[1]{\Ex}{[}{]}{}{#1}
\DeclarePairedDelimiterXPP\entropy[1]{\entropyOp}{(}{)}{}{#1}
\DeclarePairedDelimiterXPP\vectorization[1]{\vecOp}{(}{)}{}{#1}
\DeclarePairedDelimiterXPP\locc[1]{\loccOp}{(}{)}{}{#1}
\DeclarePairedDelimiterXPP\mincut[2]{\mincutOp}{(}{)}{}{#1,#2}
\DeclarePairedDelimiterXPP\stab[1]{\stabOp}{(}{)}{}{#1}
\DeclarePairedDelimiterXPP\nullity[1]{\nullityOp}{(}{)}{}{#1}

\providecommand\given{}
\newcommand\SetSymbol[1][]{%
	\nonscript\:#1\vert
	\allowbreak
	\nonscript\:
	\mathopen{}}
\DeclarePairedDelimiterX\set[1]\{\}{%
	\renewcommand\given{\SetSymbol[\delimsize]}
	#1
}

\makeatletter
\DeclareRobustCommand{\rvdots}{%
	\vbox{
	  	\baselineskip4\p@\lineskiplimit\z@
	  	\kern-\p@
	  	\hbox{.}\hbox{.}\hbox{.}
}}
\makeatother
\makeatletter
\def\Ddots{\mathinner{\mkern1mu\raise\p@
\vbox{\kern7\p@\hbox{.}}\mkern2mu
\raise4\p@\hbox{.}\mkern2mu\raise7\p@\hbox{.}\mkern1mu}}
\makeatother

\newcommand*{\idm}{\ensuremath{\mathbbm{1}}}

\usepackage{booktabs}

\usepackage[lined,algoruled,algosection,linesnumbered,noend]{algorithm2e}
\SetAlCapSkip{0.5em}
\SetEndCharOfAlgoLine{} 
\SetKwInOut{Data}{Data}
\SetKwInOut{Input}{Input}
\SetKwInOut{Output}{Output}
\SetKwFor{For}{for}{\string:}{endfor}
\SetKwFor{While}{while}{\string:}{endw}
\SetKwFor{ForEach}{foreach}{\string:}{endfch}

\SetFuncArgSty{}
\SetArgSty{}

\SetKwProg{Fn}{function}{\string:}{endfunction}
\SetKwFunction{execute}{execute}

\usepackage[%
	labelfont=bf
]{caption}
\usepackage{subcaption}

\definecolor{dark-red}{rgb}{0.4,0.15,0.15}
\definecolor{dark-blue}{rgb}{0.15,0.15,0.4}
\definecolor{medium-blue}{rgb}{0,0,0.5}

\definecolor{mycomment}{rgb}{0.3,0.7,0.8}
\definecolor{mygray}{rgb}{0.5,0.5,0.5}
\definecolor{lightgray}{rgb}{0.95,0.95,0.95}
\definecolor{mymauve}{rgb}{0.58,0,0.82}

\usepackage[unicode=true]{hyperref}
\hypersetup{%
	breaklinks=true,
	colorlinks,
	linkcolor={dark-blue},
	citecolor={dark-blue},
	urlcolor={medium-blue},
	pdfauthor = {Vadym Kliuchnikov and Eddie Schoute},
	pdfpagemode=UseOutlines,
	pdfborder={0 0 0},
	pdfcreator={},
	pdfproducer={}
}

\usepackage{enumitem}
\newlist{ienumerate}{enumerate*}{1}
\setlist*[ienumerate,1]{%
    label=(\roman*),
}

\usepackage{adjustbox} 

\usepackage{tikz}
\usetikzlibrary{quantikz2}

\usepackage[nameinlink,capitalise]{cleveref}
\crefname{figure}{Figure}{Figures}
\crefname{equation}{}{} 
\Crefname{equation}{Equation}{Equations}

\newtheorem{theorem}{Theorem}[section]
\newtheorem{lemma}[theorem]{Lemma}
\newtheorem{corollary}[theorem]{Corollary}

\theoremstyle{definition}
\newtheorem{definition}[theorem]{Definition}
\theoremstyle{remark}

\newcommand{\cnot}[0]{{\textsc{cx}}}
\newcommand{\cz}[0]{{\textsc{cz}}}

\newcommand{\cknot}[1]{\ensuremath{\textsc{c}^{#1}\kern-0.1em\textsc{not}}}

\newcommand{\sw}[0]{{\textsc{swap}}} 

\AtBeginDocument{%
  \DeclareFontShape{T1}{lmr}{m}{scit}{<->ssub*lmr/m/scsl}{}%
}

\author{Vadym Kliuchnikov}
\affil{Microsoft Quantum}
\author{Eddie Schoute\thanks{\href{mailto:eschoute@lanl.gov}{eschoute@lanl.gov}}}
\affil{Computer, Computational, and Statistical Sciences Division, Los Alamos National Laboratory}

\title{Minimal entanglement for injecting diagonal gates}
\date{}

\begin{document}

\maketitle

\begin{abstract}
	\noindent Non-Clifford gates are frequently exclusively implemented on fault-tolerant architectures by first distilling magic states in specialised magic-state factories.
	In the rest of the architecture, the \emph{computational space},
	magic states can then be consumed by a stabilizer circuit to implement non-Clifford operations.
	We show that the connectivity between the computational space and magic state factories forms a fundamental bottleneck on the rate at which non-Clifford operations can be implemented.
	We show that the \emph{nullity} of the magic state, $\nullity{\ket{D}}$ for diagonal gate $D$,
	characterizes the non-local resources required to implement $D$ in the computational space.
	As part of our proof, we construct local stabilizer circuits that use only $\nullity{\ket{D}}$ ebits to implement $D$ in the computational space
	that may be useful to reduce the non-local resources required to inject non-Clifford gates.
	Another consequence is that the edge-disjoint path compilation algorithm [\emph{PRX Quantum} 3, 020342 (2022)]
	produces minimum-depth circuits for implementing single-qubit diagonal gates.
\end{abstract}

\section{Introduction}
Clifford operations, such as single-qubit Pauli rotations and \cnot{} (\textsc{cnot}) gates,
are relatively easy to perform fault-tolerantly but do not allow for universal quantum computing.
Adding a non-Clifford gate enables universal computation.
A common method for deterministically implementing non-Clifford diagonal gates is through gate injection~\cite{Zhou2000},
which implements the gate by a simpler (e.g., Clifford) circuit and uses a \emph{magic} resource state,
\begin{equation}
	\ket{D} \coloneqq D\ket{+}^{\otimes n},
\end{equation}
for the diagonal gate $D$ in the computational basis acting on $n$ qubits~\cite{Bravyi2005,Beverland2020}.
Producing high-fidelity magic states is difficult~\cite{Bravyi2012,Jones2013,Haah2017a,Campbell2018};
it takes a significant fraction of space and time on a fault-tolerant architecture to produce sufficient magic states.

Therefore, as part of the design process, certain regions are commonly dedicated to run highly-specialised \emph{magic state factories},
whose sole task is to produce high-quality magic states~\cite{JavadiAbhari2017,Litinski2019,Beverland2022a}
We assume that magic states are produced in some region(s) on the architecture
and the rest of the architecture is the \emph{computational space}.
Non-Clifford gates will for that reason require injecting magic states from spatially separated regions.
We bound the non-local resource requirements for implementing any diagonal gate
and give an efficient, resource-optimal protocol.

The magic states required for implementing a quantum algorithm depend on its decomposition.
Most algorithmic building blocks already naturally decompose into Clifford and diagonal gates~\cite{Beverland2022a},
but any non-Clifford operation can be decomposed~\cite{Iten2016,Malvetti2021,Kliuchnikov2023}.
Diagonal gate can be implemented using Clifford operations and $\textsc{t} = \sqrt[4]{Z}$ gates,
costing one $\ket{\textsc{t}}$ state each,
but this may require communicating more qubits from the magic state factories to the computational space than necessary.
For example, the \textsc{ccz} gate may be implemented using four $\ket{\textsc{t}}$ states,
but it can be implemented using only one $\ket{\textsc{ccz}}$ state, which is supported on three qubits~\cite{Beverland2020}.
In general, we can implement any diagonal gate, $D$, using a $\ket{D}$ state, Clifford operations, and use of the $DXD^\dagger$ gate,
which is one level lower in the Clifford hierarchy~\cite{Gottesman1999}.
However, the $\ket{D}$ state is not always the most space-efficient resource state to implement $D$.
For example, we can implement a common operator in Pauli-based computation~\cite{Litinski2019} $e^{i\frac{\pi}{8} P}$,
for a Pauli string $P$ in the Pauli group $\mathcal P_n$ on $n$ qubits
(such as $Z \otimes X \otimes Z \otimes \cdots \otimes Y$),
using only a single-qubit magic state, $\ket{\textsc{t}}$---a savings of $\bigo{n}$ space over $\ket{e^{i\frac{\pi}{8}P}}$.

We show that the \emph{stabilizer nullity}~\cite{Beverland2020} characterizes the resource requirements to implement $D$.
To define the nullity,
let us define the set of \emph{stabilizers} of a state as
\begin{equation}
	\stab{\ket \psi} \coloneqq \set{P \in \mathcal P_n \given P\ket\psi = \ket\psi},
\end{equation}
then the nullity of an $n$-qubit state $\ket\psi$ is
\begin{equation}
	\nullity{\ket\psi} \coloneqq n- \log_2\abs*{\stab{\ket{\psi}}}.
\end{equation}
The nullity is non-increasing under stabilizer (Clifford) circuits, $\mathcal C$,
i.e., $\nullity{\mathcal C\ket{\psi}} \le \nullity{\ket\psi}$,
and is additive under tensor factors such that $\nullity{\ket\psi \otimes \ket{\phi}} = \nullity{\ket\psi} + \nullity{\ket\phi}$,
for any states $\ket{\psi}$ and $\ket\phi$.

We show that implementing $D$ on the computational space 
requires a stabilizer circuit that can generate at least $\nullity{\ket{D}}$ ebits of entanglement
with the magic state factories.
Generating entanglement takes time~\cite{Bravyi2006,Marien2016,Delfosse2021}
and this leads to a lower bound on the depth of the circuit needed to inject $D$~\cite{Bapat2023}.

Conversely, we give a matching efficient algorithm that uses $\nullity{\ket{D}}$ ebits between the computational space and the rest of the system to implement $D$ using local Clifford operations.
A key component is the construction of Clifford unitaries that compress the support of $D$ to $\nullity{\ket{D}}$ qubits,
more compactly encoding its action on a subspace where the stabilizers are trivial.
A previously-shown example is the correspondence between $Z$ and any Pauli string $P$ by Clifford conjugation
that allows us to implement $e^{i\frac{\pi}{8}P}$ using only one $\ket{\textsc{t}}$ state,
i.e., there exists a Clifford unitary, $C$, such that $Ce^{i\frac{\pi}{8}P} C^\dagger = \textsc{t}$.
Our algorithm can therefore decrease the communication complexity between magic state factories and the computational space.

Finally, we apply our results to compilation algorithms for fault-tolerant architectures.
As a consequence of our lower bound,
the connectivity between the computational space and magic state factories limits the rate at which magic states can be injected.
It is therefore advisable to not put all magic state factories in one place as in~\cite{Beverland2022a}
if they are not well-connected to the rest of the architecture.
More precisely, the minimum-cut that separates the computational space from the magic state factories implies a lower bound on the circuit depth for implementing diagonal gates.
Additionally, we conclude that the edge-disjoint path compilation algorithm (EDPC)~\cite{Beverland2022} finds a depth-optimal circuit implementing single-qubit diagonal gates (such as \textsc{t}).

\section{Local simulation cost}
We define preliminaries, working up to the measure of interest, the \emph{local simulation cost},
which upper bounds the the number of ebits required to locally implement a post-selected stabilizer circuit.

We define \emph{Clifford unitaries} as the products of $\set{H, S, \cnot}$,
which are the first level of the Clifford hierarchy and are not universal for quantum computation.
By adding a \textsc{t} gate, we obtain the universal gate set Clifford+\textsc{t}.
In addition to Clifford unitary operations,
a fault-tolerant architecture commonly also supports the following \emph{stabilizer operations}:
\begin{definition}[Stabilizer operation]
	A stabilizer operation is one of the following operations:
	\begin{enumerate}
		\item Clifford unitaries, potentially classically controlled.
		\item State preparation of $\ket 0$.
		\item Destructive $Z$-basis measurements.
		\item Two-qubit joint Pauli measurements.
	\end{enumerate}
\end{definition}
Our results can be generalized to other sets of stabilizer operations,
only their local simulation cost (defined later) needs to be computed.
A \emph{stabilizer circuit} is a sequence of stabilizer operations.

We frequently work with post-selected stabilizer circuits, which map pure states to unnormalized pure states.
When post-selected, every stabilizer operation is a linear map.
Therefore, a post-selected stabilizer circuit, $\mathcal C$, is a linear map,
which maps pure states to pure states,
and we can write $\mathcal C\ket\psi = \ket\phi$,
where $\ket{\psi}$ and $\ket\phi$ are (unnormalized) pure states.

We are interested in the properties of circuits that inject (non-Clifford) diagonal operators.
Inherently, there is a bipartition that separates the computational space and the space where non-Clifford diagonal gates can be applied.
We formalize this distinction by the following definition.
\begin{definition}[$(\mathcal X, \mathcal Y)$-injection circuit]
	Let $\mathcal X$ and $\mathcal Y$ be Hilbert spaces,
	then an $(\mathcal X, \mathcal Y)$-injection circuit acts on the joint Hilbert space $\mathcal X \otimes \mathcal Y$
	by stabilizer operations and by non-Clifford gates only on $\mathcal X$.
\end{definition}
A \emph{local operation} acts non-trivially exclusively on $\mathcal X$ or on $\mathcal Y$.
Any other operation in an $(\mathcal X, \mathcal Y)$-injection circuit is \emph{non-local}.
An \emph{$(\mathcal X, \mathcal Y)$-stabilizer circuit} is a special case of an $(\mathcal X, \mathcal Y)$-injection circuit
that cannot apply non-Clifford gates.

A \emph{stabilizer state} is any state that can be constructed by a stabilizer circuit without input.
As we prove below,
non-local ebits are sufficient to construct any stabilizer state using local Clifford unitaries.

\begin{lemma}\label{lem:stabilizerStateDecomposition}
	Given a stabilizer state $\ket{S}$ prepared by a post-selected $(\mathcal X, \mathcal Y)$-stabilizer circuit,
	there is an efficient algorithm to find the unique $k \in \mathbb N$
	and find local Clifford unitaries $C_{\mathcal X}$ and $C_{\mathcal Y}$ acting on $\mathcal X$ and $\mathcal Y$,
	respectively, such that
	\begin{equation}
		\begin{quantikz}
			\lstick{$\bra{0}^{\otimes n_1}$} & \qwbundle{} & \gate[2]{C_{\mathcal X}} & \\
			\lstick[2]{$\bra{\Psi_+}^{\otimes k}$} & \qwbundle{} & & \\
			& \qwbundle{} & \gate[2]{C_{\mathcal Y}} & \\
			\lstick{$\bra{0}^{\otimes n_2}$} & \qwbundle{} & &
		\end{quantikz}
		= \ket{S},
	\end{equation}
	where the Bell states $\ket{\Psi_+} = \frac{\ket{00} + \ket{11}}{\sqrt 2}$ are entangling $\mathcal X$ and $\mathcal Y$.
\end{lemma}
\begin{proof}
	See~\cite[Lemma~II.24]{Haah2017}.
	The number $k$ is unique since local unitaries cannot change the von Neumann entropy of $\ket{S}$ when one register is traced out.
\end{proof}

\cref{lem:stabilizerStateDecomposition} also allows us to define the \emph{distillable entanglement} of a stabilizer state.
\begin{definition}[Stabilizer distillable entanglement]
	Given a stabilizer state $\ket{S}$ prepared by a post-selected $(\mathcal X, \mathcal Y)$-stabilizer circuit,
	then the distillable entanglement of $\ket{S}$ is the number of ebits, $k \in \mathbb N$,
	given by \cref{lem:stabilizerStateDecomposition}.
\end{definition}

By \cref{lem:stabilizerStateDecomposition},
the stabilizer distillable entanglement is efficiently computable.
Additionally, the stabilizer distillable entanglement is equal to the distillable entanglement~\cite[Ch.~6]{WatrousBook}
since local operations cannot increase the entanglement across the bipartition of the system.
Therefore, we will simply use the term \emph{distillable entanglement}.

We may now quantify how much a post-selected $(\mathcal X,\mathcal Y)$-stabilizer circuit, $\mathcal C$,
can entangle $\mathcal X$ and $\mathcal Y$.
We obtain the \emph{Choi state of $\mathcal C$}, $\ket{\Psi_{\mathcal C}}$, by applying $\mathcal C$ to an initial state
consisting of the maximally entangled states on $\mathcal X \otimes \mathcal X$ and $\mathcal Y \otimes \mathcal Y$.
Let the maximally entangled state on the Hilbert space $\mathcal X \otimes \mathcal X$ be
\begin{equation}
	\ket{\Psi_+^{\mathcal X}} \coloneqq \sum_{i=0}^{\dim(\mathcal X)-1} \ket{i} \otimes \ket{i},
\end{equation}
then we obtain the Choi state of $\mathcal C$
\begin{equation}
	\ket{\Psi_{\mathcal C}} \propto \paren{\idm_{\mathcal X} \otimes \mathcal C \otimes \idm_{\mathcal Y}}\paren{\ket{\Psi_+^{\mathcal X}} \ket{\Psi_+^{\mathcal Y}}}
\end{equation}
by normalizing the state.

\begin{figure}
	\centering
	\begin{equation*}
		\begin{quantikz}
			\lstick[2]{$\bra{\psi}$} & \gate[2]{\mathcal C}\wire[l][1]["\mathcal X"]{a} & \\
			& \wire[l][1]["\mathcal Y"]{a} &
		\end{quantikz}
		\propto
		\begin{quantikz}
			\lstick[6]{$\bra{\Psi_{\mathcal C}}$} & \wire[l][1]["\mathcal X"]{a} & & & &\\
			& \wire[l][1]["\mathcal X"]{a} &  & \rstick[2]{$\ket{\Psi_+^{\mathcal X}}$} \\
			&  \wireoverride{n} &  \lstick[2]{$\bra{\psi}$}\wireoverride{n} & \wire[l][1]["\mathcal X"]{a}  \\
			&  \wireoverride{n} & 						   \wireoverride{n} &\wire[l][1]["\mathcal Y"]{a} \rstick[2]{$\ket{\Psi_+^{\mathcal Y}}$} \\
			&  \wire[l][1]["\mathcal Y"]{a}& &  \\
			&  \wire[l][1]["\mathcal Y"]{a} & & & &
		\end{quantikz}
		=
		\begin{quantikz}
			\lstick{$\bra{0}^{\otimes n_1}$} & \gate[2]{C_{\mathcal X}}\qwbundle{} & \wire[l][1]["\mathcal X"]{a}  & \\
			\lstick[4]{$\bra{\Psi_+}^{\otimes k}$} & \qwbundle{} & \wire[l][1]["\mathcal X"]{a} \rstick[2]{$\ket{\Psi_+^{\mathcal X}}$} \\
			&  \lstick[2]{$\bra{\psi}$}\wireoverride{n}  & \wire[l][1]["\mathcal X"]{a} \\
			&  \wireoverride{n}  & \wire[l][1]["\mathcal Y"]{a} \rstick[2]{$\ket{\Psi_+^{\mathcal Y}}$} \\
			& \gate[2]{C_{\mathcal Y}}\qwbundle{} & \wire[l][1]["\mathcal Y"]{a}  \\
			\lstick{$\bra{0}^{\otimes n_2}$} & \qwbundle{} & \wire[l][1]["\mathcal Y"]{a} &
		\end{quantikz}
	\end{equation*}
	\caption{%
		Quantifying how much a post-selected $(\mathcal X, \mathcal Y)$-stabilizer circuit, $\mathcal C$,
		can entangle across the bipartition of $\mathcal X$ and $\mathcal Y$
		by the distillable entanglement of its Choi state, $\ket{\Psi_{\mathcal C}}$.
		This Choi state can be used to implement $\mathcal C$ on any $\ket{\psi} \in \mathcal X \otimes \mathcal Y$
		through local stabilizer operations and post-selection on
		$\ket{\Psi_+^{\mathcal X}} \coloneqq \frac{1}{\sqrt{\dim \mathcal X}} \sum_{i=1}^{\dim \mathcal X} \ket{i}\ket{i}$
		(and similarly for $\ket{\Psi_+^{\mathcal Y}}$),
		effectively implementing a qudit teleportation without corrections.
		By \cref{lem:stabilizerStateDecomposition}, $\ket{\Psi_{\mathcal C}}$ can be decomposed into local Clifford unitaries acting on $k$ Bell pairs and ancillas,
		quantifying how much entanglement is necessary to simulate $\mathcal C$ for any input state.
	}
	\label{fig:choiInjection}
\end{figure}
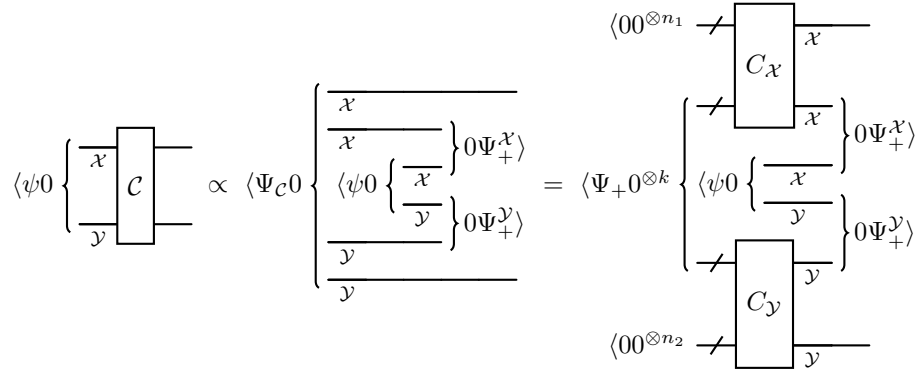

The distillable entanglement of $\ket{\Psi_{\mathcal C}}$ quantifies how much entanglement $\mathcal C$ can generate.
We showed a post-selected $(\mathcal X^{\otimes 2}, \mathcal Y^{\otimes 2})$-stabilizer circuit that uses
$\mathcal C$ and local operations to prepare $\ket{\Psi_{\mathcal C}}$.
Conversely, we show below that we can apply $\mathcal C$ to an arbitrary state
by acting on $\ket{\Psi_{\mathcal C}}$
using only local stabilizer operations in a post-selected $(\mathcal X^{\otimes 3}, \mathcal Y^{\otimes 3})$-stabilizer circuit.
Therefore, $\ket{\Psi_{\mathcal C}}$ is locally equivalent to implementing $\mathcal C$
and its distillable entanglement characterizes how many
non-local ebits are needed to implement $\mathcal C$ in the worst case (preparing the Choi state).
See \cref{fig:choiInjection} for a circuit diagram describing the procedure.

\begin{lemma}\label{lem:localSimulation}
	Given the Choi state, $\ket{\Psi_{\mathcal C}}$, of a post-selected $(\mathcal X, \mathcal Y)$-stabilizer circuit $\mathcal C$,
	there exists a post-selected $(\mathcal X^{\otimes 3}, \mathcal Y^{\otimes 3})$-stabilizer circuit $\Theta$
	consisting of local stabilizer operations
	such that
	\begin{equation}
		\Theta \ket{\psi}\ket{\Psi_{\mathcal C}} \propto \mathcal C \ket\psi,
	\end{equation}
	for any state $\ket{\psi} \in \mathcal X \otimes \mathcal Y$.
\end{lemma}
\begin{proof}
	We specify $\Theta$ in \cref{fig:choiInjection}.
\end{proof}

For our purposes it is sufficient to only consider post-selected stabilizer circuits,
but \cref{lem:localSimulation} can be extended to not require post-selection in some cases.
For example, any stabilizer circuit, $\mathcal C$, with only Pauli unitaries controlled by parities of measurement outcomes
also has an equivalent implementation that uses a state $\ket{\Psi_{\mathcal C}}$ followed by Bell measurements and conditional Pauli unitaries.
The appropriate Pauli correction and the map between measurement outcomes
can be recovered using a stabilizer circuit verification algorithm~\cite{Kliuchnikov2023a}.

We now define a measure called the \emph{local simulation cost} that upper bounds the Bell pair resources required to locally implement a post-selected stabilizer circuit.

\begin{definition}[Local simulation cost]\label{def:localSimulationCost}
	Consider the post-selected stabilizer operation (with fixed classical control), $S$,
	acting on $\mathcal X \otimes \mathcal Y$.
	We define the local simulation cost of $S$ to be the distillable entanglement of
	its Choi state, 
	$\ket{\Psi_S}$,
	where we interpret $S$ as a post-selected $(\mathcal X, \mathcal Y)$-stabilizer circuit.

	The local simulation cost of a post-selected $(\mathcal X, \mathcal Y)$-injection circuit is the sum of the local simulation costs of all its stabilizer operations.
\end{definition}

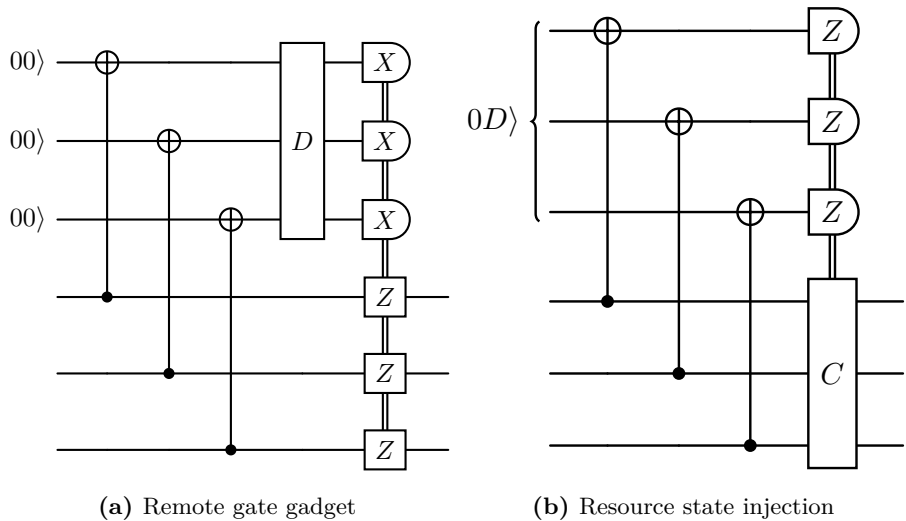
\begin{figure}
	\centering
	\begin{subfigure}{0.5\textwidth}
		\centering
		\begin{adjustbox}{width=\textwidth}
			\begin{quantikz}
				\lstick{$\ket{0}$} & \targ{} &&& \gate[3]{D} &\meterD{X} \wire[d][5]{c} \\
				\lstick{$\ket{0}$}		& &\targ{} && &\meterD{X} \\
				\lstick{$\ket{0}$} & &&\targ{} & &\meterD{X} \\
				& \ctrl{-3} &&&& \gate{Z} & \\
				&& \ctrl{-3} &&& \gate{Z} & \\
				&&& \ctrl{-3} && \gate{Z} &
			\end{quantikz}
		\end{adjustbox}
		\caption{%
			Remote gate gadget
		}
		\label{fig:remoteGate}
	\end{subfigure}%
	\hfill
	\begin{subfigure}{0.5\textwidth}
		\centering
		\begin{adjustbox}{width=\textwidth}
			\begin{quantikz}
				\lstick[3]{$\ket{D}$} & \targ{} &&& \meterD{Z} \wire[d][3]{c} \\
				& &\targ{} && \meterD{Z} \\
				& &&\targ{} & \meterD{Z} \\
				& \ctrl{-3} &&& \gate[3]{C} & \\
				&& \ctrl{-3} && &\\
				&&& \ctrl{-3} & &
			\end{quantikz}
		\end{adjustbox}
		\caption{%
			Resource state injection
		}
		\label{fig:resourceInjection}
	\end{subfigure}
	\caption{%
		Two ways of applying a diagonal gate $D$ to an arbitrary state.
		Both generalize in the obvious way to arbitrarily-sized support of $D$.
		(\subref{fig:remoteGate}) The remote gate gadget~\cite{Beverland2022} can be used to defer execution of $D$,
		since only a Clifford correction remains on the system register.
		This can also allow execution of $D$ in a different region of the architecture where $D$ can be performed.
		(\subref{fig:resourceInjection}) Applying $D$ of finite Clifford hierarchy level $k$ using a resource state $\ket{D} = D\ket{+}^{\otimes 3}$
		and corrections, $C$, of Clifford hierarchy level $k-1$~\cite{Beverland2020}.
		In both cases, if the measurement outcome is all $0$, then no correction is required.
		When the ancilla space is far from the input on the architecture,
		we can use a shared ebit to implement each non-local \cnot{} using LOCC instead~\cite{Beverland2022}.
	}
\end{figure}

The local simulation cost of a post-selected $(\mathcal X, \mathcal Y)$-injection circuit, $\mathcal C$,
upper bounds the ebits required to implement $\mathcal C$ using local operations.
There are only two types of non-local operations, with a non-zero local simulation cost of 1  (see \cref{app:localSimulationCost}):
\cnot{}s and two-qubit joint Pauli measurements~\cite[Figs.\ 17 and 19]{Haner2022}.

\section{Nullity characterizes local simulation cost}
In this section, we prove that the local simulation cost of
a diagonal gate $D$ is exactly equal to $\nullity{\ket{D}}$.
We prove this result by first showing the necessary direction
and then showing that it is sufficient to have $\nullity{\ket{D}}$ ebits to implement $D$.

\begin{lemma}\label{lem:necessaryDiagonal}
	For an $(\mathcal X, \mathcal Y)$-injection circuit, $\mathcal C$,
	post-selected on any measurement outcome $m=m_1\dots m_\ell$,
	it is necessary to have local simulation cost of at least $\nullity{\ket{D}}$
	to implement a diagonal operation $D$ in $\mathcal Y$.
\end{lemma}
\begin{proof}
	Since the measurement outcomes are fixed to $m$, any controlled operations are fixed.
	The circuit works for any input state,
	so we may assume the input state is $\ket{+}^{\otimes \ell} \in \mathcal Y$
	and maps the input to $D\ket{+}^{\otimes \ell} = \ket{D}$.

	\begin{figure}
		\begin{subfigure}{0.5\textwidth}
			\centering
			\begin{adjustbox}{width=\textwidth}
				\begin{quantikz}
					& \lstick{$\bra{0}^{\otimes k_1}$}\wireoverride{n}              & \targ{} \qwbundle{}                & \gate{D_1} & \rstick{$\ket{+}^{\otimes k_1}$} \\
					\lstick[2]{$\mathcal X$} & \gate[3]{\mathcal C_1}\wireoverride{n}       & \ctrl{-1}                       & \gate[3]{\mathcal C_2} \\
					& \wireoverride{n}               & \qwbundle{}                      &               \\
					\lstick{$\bra{+}^{\otimes \ell}$}\wireoverride{n}  & \qwbundle{}                    &             &                                     &                                         \\
				\end{quantikz}
			\end{adjustbox}
			\caption{Deferring $D_1$ execution}
			\label{fig:deferredStatePreparation}
		\end{subfigure}
		\begin{subfigure}{0.5\textwidth}
			\centering
			\begin{adjustbox}{width=\textwidth}
				\begin{quantikz}
					\lstick{$\bra{0}_a^{\otimes n_1}$} & \gate[2]{C_1}\qwbundle{} & \gate[2]{\bigotimes_i D_i} & \rstick{$\ket{+}^{\otimes n_1}$} \\
					\lstick[2]{$\bra{\Psi^+}_{bc}^{\otimes k}$} &  \qwbundle{} & & \rstick{$\ket{+}^{\otimes k}$} \\
					&  \gate[2]{C_2}\qwbundle{} & &   \\
					\lstick{$\bra{0}_d^{\otimes n_2}$} & \qwbundle{} & &
				\end{quantikz}
			\end{adjustbox}
			\caption{Injecting $D$ using entanglement}
			\label{fig:localStatePreparation}
		\end{subfigure}
		\caption{%
			We can transform an $(\mathcal X, \mathcal Y)$-injection circuit that implements a diagonal operation $D$ on the $\mathcal Y$ space.
			We may apply this circuit on $\ket{+}^{\otimes \ell}$ so that the resulting state is $\ket{D}$.
			Any such circuit can be broken apart into a $(\mathcal X, \mathcal Y)$-injection circuit $\mathcal C_1$, followed by a (non-Clifford) diagonal gate $D_1$ acting on $k_1$ qubits, then an $(\mathcal X, \mathcal Y)$-injection circuit $\mathcal C_2$.
			(\subref{fig:deferredStatePreparation}) Using the remote gate gadget (\cref{fig:remoteGate})
			where we post-select on $\ket{+}^{\otimes k_1}$,
			we defer the execution of $D_1$.
			By iterating this process for $\mathcal C_1$ and $\mathcal C_2$,
			we may defer execution of all non-Clifford gates until the end of the circuit.
			(\subref{fig:localStatePreparation}) By \cref{lem:stabilizerStateDecomposition},
			we now may re-express the stabilizer state before non-Clifford gates as
			$(C_1 \otimes C_2)\ket{0}^{\otimes n_1}\ket{\Psi^+}^{\otimes k}\ket{0}^{\otimes n_2}$,
			where $C_1$ and $C_2$ are Clifford unitaries that only act locally.
		}
	\end{figure}
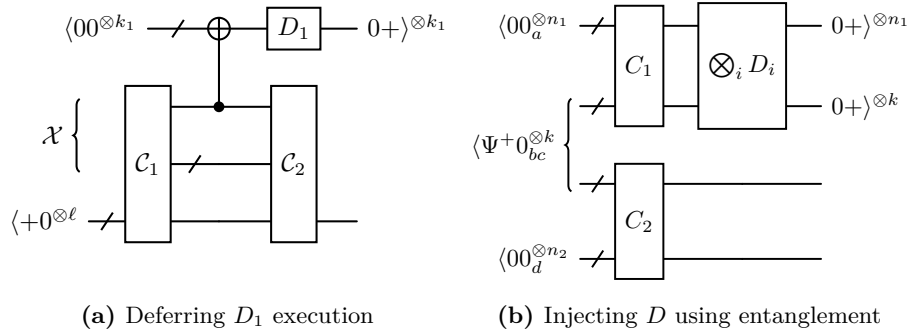

	We will transform to an equivalent circuit through only local operations.
	Therefore, the transformation cannot increase the local simulation cost of the circuit.
	We replace all instances of non-Clifford diagonal operations
	by remote gate gadgets (\cref{fig:deferredStatePreparation}).
	Since we are operating in a post-selected setting,
	we may post-select the correction outcomes to be such that no correction is necessary.
	Now we note that all non-Clifford diagonal operations commute with the rest of the circuit,
	and may therefore be performed at the end.
	Everything before these operations (including on $\mathcal Y$) is a post-selected $(\mathcal X, \mathcal Y)$-stabilizer circuit,
	and can therefore be seen as preparing a stabilizer state $\ket S$.

	We re-express $\ket S = \paren{C_1 \otimes C_2} \ket{0}^{\otimes n_1}\ket{\Psi_+}^{\otimes k}\ket{0}^{\otimes n_2}$
	using~\cref{lem:stabilizerStateDecomposition} for $n_1,k,n_2 \in \mathbb N$
	and $C_1, C_2$ Clifford unitaries (\cref{fig:localStatePreparation}).
	Let all deferred non-Clifford diagonal gates be $D_i$,
	then we define the pure state
	\begin{equation}
		\ket \psi \coloneqq \bra{+}_{ab}^{\otimes (n_1 + k)} \bigl[(\bigotimes_i D_i) C_1\bigr]_{ab}\ket{0}_a^{\otimes n_1}\ket{\Psi_+}_{bc}^{\otimes k},
	\end{equation}
	which is prepared on a subspace of $\mathcal Y$.
	We use the fact that $C_2$ is a Clifford unitary to derive
	\begin{equation}
		\nullity{\ket{D}} = \nullity*{\mathcal C \ket{+}^{\otimes \ell}} = \nullity*{C_2 \ket{\psi}\ket{0}^{\otimes n_2}} = \nullity{\ket{\psi}} \le k
	\end{equation}
	since the nullity is additive under tensor products, $\nullity{\ket{0}} = 0$,
	and the fact that the nullity is at most the number of qubits in a state.
	Thus $\mathcal C$ must have local simulation cost of at least $\nullity{\ket{D}}$.
\end{proof}

We now prove that $\nullity{\ket{D}}$ ebits are sufficient to implement $D$ using local operations.
The main difficulty is when $\nullity{\ket{D}} < n$, the number of qubits supporting $D$,
because it $D$ can be implemented with $n$ ebits using a remote gate gadget (\cref{fig:remoteGate}).

\begin{lemma}\label{lem:compressDiagonal}
	For any diagonal gate $D_n$ supported on $n$ qubits, there exist Clifford unitaries $V$ and $V'$ such that
	\begin{equation}
		V D_n V' = \idm_{n-n'} \otimes D_{n'},
	\end{equation}
	for $n' \coloneqq \nullity{\ket{D_n}}$ and a diagonal operator $D_{n'}$ supported on $n'$ qubits.
\end{lemma}

\begin{proof}
	We prove this Lemma by induction on $n'$.
	For the base case, we are given an $n$-qubit operator with $\nullity{\ket{D_n}} = n$,
	then $V = V' = \idm$ and the result follows trivially.

	For the induction step, we are given a $k$-qubit operator $D_k$, with $\nullity{\ket{D_k}} = n' < k$.
	Therefore, by definition of the nullity,
	there must be a $+1$-expectation value of some stabilizer generator $X^bZ^c$ such that
	\begin{equation}
		\abs*{\bra{D_k} X^bZ^c \ket{D_k}}
		= \abs*{\bra{+}^n D_k^\dagger X^bZ^c D_k\ket{+}^n}
		= 1,
	\end{equation}
	for bitstrings $b, c \in \set{0,1}^k$ and using the notation $P^b = P_1^{b_1}\cdots P_k^{b_k}$ for single-qubit operator $P$.
	We can see that there is some $b_x = 1$ for $x \in [n]$,
	since
	\begin{equation}
		\abs*{\bra{D_k} Z^c \ket{D_k}} = \abs*{\bra{+}^n Z^c \ket{+}^n} = 0
	\end{equation}
	for any $c$ with $\abs{c} > 0$.

	Given the following facts on the \cnot{} gate
	\begin{equation}\label{eq:cxEqualities}
		\cnot{}_{ij} Z_j \cnot{}_{ij} = Z_i Z_j, \qquad \cnot{}_{ij} X_i \cnot{}_{ij} = X_i X_j,
	\end{equation}
	we can find a diagonal operator that has a $+1$ expectation with one $X$ operator.
	We define
	\begin{equation}
		V_X \coloneqq \prod_{j \in [n] \setminus \set{x} : b_j = 1} \cnot{}_{xj}
	\end{equation}
	and then it is easy to see, using \cref{eq:cxEqualities}, that the diagonal operator $D_X \coloneqq V_X D_k V_X$
	has expectation value
	\begin{equation}
		1 = \abs*{\bra{D_k} X^b Z^c \ket{D_k}} = \abs*{\bra{D_X} V_X X^b V_X Z^{c'} \ket{D_X}} = \abs*{\bra{D_X} X_x Z^{c'} \ket{D_X}},
	\end{equation}
	for some bitstring $c' \in \set{0,1}^k$.

	Now we find a diagonal operator that has $+1$ expectation without $Z$ operators.
	We use the two facts
	\begin{equation}
		\cz{}_{ij} X_i \cz{}_{ij} = X_i Z_j
	\end{equation}
	and
	\begin{equation}
		S X S^\dagger = -Y =  i X Z.
	\end{equation}
	We define
	\begin{equation}
		V_Z \coloneqq S_x^{c'_x} \prod_{j \in [n] \setminus \set{x} : c'_j = 1} \cz{}_{xj}
	\end{equation}
	such that the diagonal operator $D_Z \coloneqq V_Z D_X$ has expectation value
	\begin{multline}
		1 = \abs*{\bra{D_X} X_x Z^{c'} \ket{D_X}} = \abs*{\bra{D_Z} V_Z X_x Z^{c'} V_Z^\dagger \ket{D_Z}}\\
		= \abs*{\bra{D_Z} S_x^{c_x} X_x Z_x^{c_x} (S_x^\dagger)^{c_x} \ket{D_Z}} = \abs*{\bra{D_Z} X_x \ket{D_Z}}.
	\end{multline}
	If $x \neq 1$, define $V_S \coloneqq \sw{}_{x1}$, otherwise define $V_S \coloneqq \idm$.
	Then, the diagonal operator $D_S \coloneqq V_S D_Z V_S$ has expectation value
	\begin{equation}\label{eq:x1Expectation}
		1 = \abs*{\bra{D_Z} X_x \ket{D_Z}} = \abs*{\bra{D_S} X_1 \ket{D_S}}.
	\end{equation}

	We claim that there exists a diagonal operator $D_{k-1}$ acting on $k-1$ qubits
	such that we can factorize $D_S = Z^b \otimes D_{k-1}$, where $b \in \set{0,1}$.
	There are two cases since $\brakett{D_S}{X_1}{D_S} \in \set{-1,1}$.
	In the $-1$ case, we define $V_+ \coloneqq Z_1$,
	then the diagonal operator $D_+ \coloneqq V_+ D_S$ has expectation
	\begin{equation}
		-1 = \brakett{D_S}{X_1}{D_S} = -\brakett{D_+}{X_1}{D_+}.
	\end{equation}
	Otherwise, we set $V_+ = \idm$.

	Now $\ket{D_+}$ has a $+1$ expectation value for $X_1$,
	so $D_+\ket{j} = X_1 D_+\ket{j}$ for any computational basis state $\ket{j}$.
	Therefore, by expanding the action of $D_+$, we obtain $X_1 D_+ \ket{j} = D_+X_1\ket{j}$
	and see that $X_1$ commutes with $D_+$.
	Clearly, $Z_1$ also commutes with the diagonal operator $D_+$
	so $D_+$ commutes with any unitary on the first qubit,
	which is a linear combination of $X_1$, $Z_1$, and $X_1Z_1$.
	The only operator that commutes with any unitary is the identity operation,
	so we can factorize $D_+ = \idm_1 \otimes D_{k-1}$.
	Finally, $\idm_1 \otimes D_{k-1} = V_+ V_S V_Z V_X D_k V_XV_S$
	and, by induction, $V$ and $V'$ can always be constructed.
\end{proof}

\begin{corollary}\label{cor:sufficientDiagonal}
	There exists an $(\mathcal X, \mathcal Y)$-injection circuit
	that implement a diagonal gate $D$ on $\mathcal Y$ using only local operations
	and $\nullity{\ket{D}}$ ebits between $\mathcal X$ and $\mathcal Y$.
\end{corollary}
\begin{proof}
	By \cref{lem:compressDiagonal},
	there exists Clifford unitaries $V, V'$ such that $VDV' = \idm_{n-n'} \otimes D_{n'}$,
	for $n' \coloneqq \nullity{\ket{D}}$.
	Clifford unitaries preserve nullity,
	so
	\begin{equation}
		\nullity{\ket{D}} = \nullity{\ket{VDV'}} = \nullity{\ket{D_{n'}}}
	\end{equation}
	since the nullity is additive under tensor factors and $\nullity{\ket{\idm}} = 0$.
	We thus can apply $D_{n'}$ by a remote gate gadget (\cref{fig:remoteGate}) using $n'$ ebits,
	apply $V^\dagger$ and $(V')^\dagger$ locally, and obtain $D$ as required.
\end{proof}

We now prove the main theorem of this section.
\begin{theorem}\label{thm:diagonalSimulationCost}
	For any $(\mathcal X, \mathcal Y)$-injection circuit,
	it is necessary and sufficient to have local simulation cost $\nullity{\ket{D}}$
	to implement a diagonal gate $D$ on $\mathcal Y$.
\end{theorem}
\begin{proof}
	\cref{lem:necessaryDiagonal} proves the necessary part.
	It is sufficient since $\nullity{\ket{D}}$ \cnot{}s can create as many Bell pairs,
	following which \cref{cor:sufficientDiagonal} shows how to implement $D$ using local operations.
	This has a non-local simulation cost of $\nullity{\ket{D}}$, as required.
\end{proof}

\section{Circuit compilation}
\begin{table}
	\centering
	\begin{tabular}{lr} \toprule
		State                                                                         & Nullity                \\
		\midrule
		$\ket{e^{i\theta Z}} : \theta \in \frac{\pi}{4}\mathbb Z$                     & 0                      \\
		$\ket{e^{i\theta Z}} : \theta \in \mathbb R \setminus \frac{\pi}{4}\mathbb Z$ & 1                      \\
		$\ket{\textsc{cs}}$                                                    & 2                      \\
		$\ket{\textsc{c}^{n-1}\textsc{z}} : n \ge 3$                    & n~\cite{Beverland2020} \\
		\bottomrule
	\end{tabular}
	\caption{%
		The stabilizer nullity of common resource states of diagonal gates.
		We define $S = \sqrt{Z}$, then $\textsc{cs}$ is the controlled version.
	}
	\label{tab:nullities}
\end{table}

We apply the results of the previous section to compilation problems in fault-tolerant architectures.
To be able to give circuit lower bounds, only nullities of diagonal gates are required,
so we give the nullities of a few common diagonal gates in~\cref{tab:nullities}.
Furthermore, in this section we assume that we are given a simple graph $G$
whose vertices, $V(G)$, represent qubit systems
and whose edges specify the available stabilizer operation interactions,
we define $X \subseteq V(G)$ to correspond with a Hilbert space $\mathcal X$
and $Y \subseteq V(G)$ to correspond with a Hilbert space $\mathcal Y$.
The minimum size edge-cut that separates $X$ from $Y$ is denoted $\mincut{X}{Y}$.
Now we show that the nullity of non-Clifford diagonal gates implies a circuit lower bound
based on the available stabilizer interactions on the architecture.

\begin{theorem}\label{thm:rotationsLB}
	An $(\mathcal X, \mathcal Y)$-injection circuit needs a depth of at least
	\begin{equation}
		\sum_{i=1}^k  \frac{\nullity{\ket{g_i}}}{\mincut{X}{Y}}
	\end{equation}
	to implement a set of parallel diagonal gates, $\mathcal G = \set{g_1, \dots, g_k}$, supported on $\mathcal Y$.
\end{theorem}
\begin{proof}
	By \cref{lem:necessaryDiagonal}, the local simulation cost to implement $\mathcal G$
	with an $(\mathcal X, \mathcal Y)$-injection circuit is at least $m \coloneqq \sum_{i=1}^k \nullity{\ket{g_i}}$.
	Consider any partitioning of $G$, $(Z, \bar Z)$,
	such that $Y \subseteq Z$ and $X \subseteq \bar Z$.
	Since a stabilizer operation that interacts $Z$ with $\bar Z$ can only increase the local simulation cost by 1 per time step,
	the circuit depth is lower bounded by $\frac{m}{\abs{E(Z, \bar Z)}}$
	where $E(Z, \bar Z)$ is the set of edges between $Z$ and $\bar Z$.
	Let us maximize over $Z$ to see that the circuit depth is lower bounded by
	\begin{equation}
		\max_{Z: Y \subseteq Z, X\subseteq \bar Z} \frac{m}{E(Z, \bar Z)} = \frac{m}{\mincut{X}{Y}}
	\end{equation}
	as required.
\end{proof}

We show that the edge-disjoint path compilation (EDPC) algorithm~\cite{Beverland2022}
compiles single-qubit diagonal gates to a depth-optimal injection circuit.

\begin{theorem}\label{thm:edpcRotationsOptimal}
	EDPC produces an $(\mathcal X, \mathcal Y)$-injection circuit
	of (logical) depth at most
	\begin{equation}
		4\sum_{i=1}^k\frac{\nullity{\ket{g_i}}}{\mincut{X}{Y}}
	\end{equation}
	to implement a set of parallel single-qubit diagonal gates, $\mathcal G = \set{g_1, \dots, g_k}$, supported on $\mathcal Y$.
	This is the asymptotically optimal circuit depth to implement $\mathcal G$.
\end{theorem}
\begin{proof}
	Let $m \coloneqq \sum_{i=1}^k \nullity{\ket{g_i}}$,
	then $m$ counts the number of non-Clifford gates in $\mathcal G$.
	Any Clifford operation $g \in \mathcal G$ is immediately implemented directly on the qubit by EDPC.
	For the remainder, we may assume $\mathcal G$ consists only of non-Clifford operations.

	A straightforward modification of the maximum flow problem solved in EDPC
	generalizes the boundary vertices to an arbitrary region $X$.
	Now, EDPC finds a maximum flow solution between $X$ and $Y$ that implies a set of edge-disjoint paths along each of which
	EDPC is able to implement long-range \cnot{} operations at a (logical) depth of at most $4$.
	Each such \cnot{}, together with a local remote gate gadget (\cref{fig:remoteGate}) implements a gate $g \in \mathcal G$.
	By the min-cut max-flow theorem~\cite{Ford1956},
	we know that the maximum flow found by EDPC that connects $Y$ to $X$ equals the $\mincut{X}{Y}$.
	Therefore, EDPC produces a compiled circuit of depth at most $4\frac{m}{\mincut{X}{Y}}$ that implements all gates in $\mathcal G$.
	\cref{thm:rotationsLB} provides an asymptotically matching lower bound.
\end{proof}

Note that it is easy to generalize \cref{thm:edpcRotationsOptimal} to gate sets $\mathcal G$
without the parallel requirement since we can simply multiply together sequences of single-qubit diagonal gates.

\section{Conclusion}
We considered a bipartite system,
where magic state factories produce magic states on one side and the computational space consumes magic states to perform (non-Clifford) diagonal operations, $D$.
We showed that any implementation of $D$ needs a circuit with a local simulation cost of at least $\nullity{\ket{D}}$.
Then we gave an algorithm to construct Clifford unitaries $V$ and $V'$ given $D$ so that $VDV' = \idm_{n-n'} \otimes D'$,
for $n' = \nullity{\ket{D}}$ and a diagonal unitary $D'$ with support on $n'$ qubits.
Using $D'$, we can perform a remote gate gadget using only $\nullity{\ket{D}}$ ebits to apply $D$ with local stabilizer circuits.
Since $\nullity{\ket{D}}$ ebits can be constructed by a circuit with the same local simulation cost,
we have characterized the local simulation cost of $D$.

When laying out fault-tolerant quantum computations, our results show that the computational space needs to be well-connected to the magic state factories,
depending on the desired consumption rate of magic states.
In Pauli-based computation~\cite{Litinski2019} there is little need for more than a constant rate of $\ket{\textsc{t}}$ gate consumption.
Modified versions of Pauli-based computation~\cite[Sec.\ 5.1]{Litinski2019} can be used to implement more \textsc{t} gates in parallel,
but come with significant increases in total space-time cost when compared to standard Pauli-based computation~\cite[Sec.\ V.A]{Chamberland2022}.
However, when some parallelism of non-Clifford diagonal gates is maintained during layout,
a higher rate of magic state consumption is required for fast execution.
To avoid a bottleneck in transporting magic states,
we show that it is necessary to spread out the magic state factories
or interleave magic state factories with computation as in~\cite{Gidney2019}.
Finally, we show that the edge-disjoint path compilation algorithm (EDPC)~\cite{Beverland2022}
produces an asymptotically minimum-depth circuit for compiling single-qubit diagonal gates.

\section*{Acknowledgements}
V.K.\ and E.S.\ contributed to this work equally.
E.S.\ was supported by the U.S.\ Department of Energy, Office of Science,
National Quantum Information Science Research Centers, Quantum Science Center.

\printbibliography%

\appendix

\section{Local simulation cost of interacting stabilizer operations}\label{app:localSimulationCost}
We explicitly compute the local simulation cost of the two stabilizer operations that
can interact non-locally: \cnot{} and a two-qubit Pauli measurement.
The local simulation cost of stabilizer operations is defined to be the distillable entanglement of its Choi state (\cref{def:localSimulationCost}).

We first compute the distillable entanglement of $\ket{\Psi_\cnot{}}$.
Recall that $\cnot{}$ acts on Pauli unitaries by conjugation as follows
\begin{align}
	\cnot{}_{1,2}(X_{1})      & =X_{1}X_{2} \\
	\cnot{}_{1,2}(X_{2})      & =X_{2}      \\
	\cnot{}_{1,2}(Z_{1})      & =Z_{1}      \\
	\cnot{}_{1,2}(Z_{2})      & =Z_{1}Z_{2},
\end{align}
and that the Bell state, $\ket{\Psi_+} = \frac{1}{\sqrt 2} \paren{\ket{00} + \ket{11}}$,
has stabilizers
\begin{equation}
	\stab{\ket{\Psi_+}} = \set{\idm, XX, ZZ, -YY} = \left< XX,ZZ \right>,
\end{equation}
for $\left<G\right>$ the group generated by the elements in $G$.
Now consider two Bell states on qubits $r_{1},r_{2},t_{1},t_{2}$
where $r$ stands for reference and $t$ stands for target.
The states have stabilizers
\begin{align}
	 & \stab{\ket{\Psi_+}_{r_{1},t_{1}} \ket{\Psi_+}_{r_{2},t_{2}}}                         \\
	 & = \left< X_{r_1}X_{t_1}, Z_{r_1}Z_{t_1}, X_{r_2}X_{t_2}, Z_{r_2}Z_{t_2}\right> \\
	 & = \left[\begin{array}{cc|cc}
			           r_{1}    & t_{1} & r_{2} & t_{2} \\
			           \hline X & X                     \\
			           Z        & Z                     \\
			                    &       & X     & X     \\
			                    &       & Z     & Z
		           \end{array}\right],
\end{align}
where we wrote the generators of the stabilizer in a table format.
Applying $\cnot{}_{t_{1},t_{2}}$ to qubits $t_{1},t_{2}$ changes the stabilizers to
\begin{equation}
	\left[\begin{array}{cc|cc}
			r_{1}    & t_{1} & r_{2} & t_{2} \\
			\hline X & X     &       & X     \\
			Z        & Z                     \\
			         &       & X     & X     \\
			         & Z     & Z     & Z
		\end{array}\right]
\end{equation}
To simplify, we the local gates $\cnot{}_{t_{1},r_{1}}$ and $\cnot{}_{r_{2},t_{2}}$ to get the stabilizers
\begin{equation}
	\left[\begin{array}{cc|cc}
			r_{1}  & t_{1} & r_{2} & t_{2} \\
			\hline & X     &       & X     \\
			Z      &                       \\
			       &       & X             \\
			       & Z     &       & Z
		\end{array}\right],
\end{equation}
which is $\ket{\Psi_+}_{t_{1},t_{2}}\ket{0}_{r_{1}}\ket{+}_{r_{2}}$,
so the local simulation cost of $\cnot{}$ is one.

Now let us consider measuring $ZZ$ on $t_{1}t_{2}$.
All other bipartite Pauli measurements have the same local simulation cost
since they can be implemented by local operations before applying a non-local $ZZ$ measurement.
We first rewrite the stabilizer generators of $\ket{\Psi_+}_{r_{1},t_{1}}\ket{\Psi_+}_{r_{2},t_{2}}$ to
\begin{equation}
	\left[\begin{array}{cc|cc}
			r_{1}    & t_{1} & r_{2} & t_{2} \\
			\hline X & X     & X     & X     \\
			Z        & Z                     \\
			         &       & X     & X     \\
			         &       & Z     & Z
		\end{array}\right]
\end{equation}
so only one Pauli anti-commutes with $Z_{t_{1}}Z_{t_{2}}$.
Measuring $Z_{t_{1}}Z_{t_{2}}$ with $+1$ outcome turns the state
into
\begin{equation}
	\left[\begin{array}{cc|cc}
			r_{1}  & t_{1} & r_{2} & t_{2} \\
			\hline & X     & X             \\
			Z      &                       \\
			       & Z     & Z             \\
			       &       &       & Z
		\end{array}\right].
\end{equation}
Now applying $\cnot{}_{t,r_{1}}$ and $\cnot{}_{r_{2},t_{2}}$ turns above stabilizer
into $\ket{\Psi_+}_{t_{1},r_{2}} \ket{0}_{r_{1}} \ket{0}_{t_{2}}$,
so the local simulation cost of a $ZZ$ measurement is also 1.

\end{document}